\documentclass[12pt, onecolumn, final]{article}
\usepackage{amsmath,amsthm,amssymb}
\usepackage{graphicx}
\usepackage{wrapfig}
\usepackage{caption}
\usepackage{subcaption,enumerate, color}
\usepackage[margin=1.15in]{geometry}
\usepackage{float}
\usepackage{multirow}
\usepackage[dvips]{epsfig}
\usepackage{lingmacros}
\usepackage{amsfonts}
\usepackage{bbm}
\usepackage[utf8]{inputenc}
\usepackage{authblk}

\usepackage{multicol}
\usepackage{url}

\theoremstyle{definition}
\newtheorem{defn}{Definition}

\newtheorem{exmp}{Example}
\newtheorem{theorem}{Theorem}
\newtheorem{lm}{Lemma}
\newtheorem{prop}{Proposition}
\newtheorem{cor}{Corollary}

\newcommand{\F}{{\mathbb F}}

\newcommand{\rmv}[1]{}

\DeclareMathOperator{\indeg}{indeg}
\DeclareMathOperator{\outdeg}{outdeg}
\DeclareMathOperator{\In}{In}
\DeclareMathOperator{\Out}{Out}

\definecolor{OliveGreen}{rgb}{0.0, 0.6, 0.0}

\title{Codes for distributed storage from 3–regular graphs\thanks{ The authors were partially supported by  the National Science Foundation
    under grants DMS-1403062, DMS-1547399 and CCF-1407623.}
}

\author{Shuhong Gao}
\author{Fiona Knoll}
\author{Felice Manganiello}
\author{Gretchen Matthews}
\affil{Clemson University}

\begin{document}

\maketitle

\begin{abstract}
  This paper considers distributed storage systems (DSSs) from a
  graph theoretic perspective. A DSS is constructed by means of the
  path decomposition of a $3$-regular graph into $P_4$ paths. The
  paths represent the disks of the DSS and the edges of
  the graph act as the blocks of storage. We deduce the properties of the DSS from a
  related graph and show their optimality.
\end{abstract}



\section{Introduction}

With the rapid increase of cloud storage, the demand for large-scale data storage rises. However, for both solid-state drives (SSD) and hard disk drives (HDD), the lifetime of the drive is variable. In many cases, if one drive fails, multiple drives fail. In the case of a disk failure, it is desirable to both replace the disk and recover the lost data. 

A distributed storage system (DSS) is a collection of $n$ disks such that the encoded files can be recovered by any $k$ of the disks, $k<n$.  
Given a disk failure, the system repairs the failure by accessing the lost data from other drives and writing the data on the replacement disk, called the newcomer. There are three types of repairs: exact, functional, and exact repair of systematic parts. Exact repair occurs when the data is exactly recovered, whereas functional repair only requires that the encoded files can still be obtained from any $k$ of the $n$ disks. Consequently, with functional repair, the recovered data may not be the same as the original. Exact repair of systematic parts is a combination of functional and exact repair;  it requires the system to contain a copy of the uncoded data and exactly repairs the uncoded data while allowing for functional repair of all other data. 

Recovery of lost data can be achieved by the use of redundancy bits. One form of redundancy consists of replicating the data. Another form involves erasure codes, which given a fixed number of redundancy bits allows for more failures than replication. There are two failure cases that we consider: bit failures within a disk and whole disk failures. The repairing of the whole disk requires external communication, meaning communication between disks. To repair bit failures within a disk, both external and internal communication, meaning communication within the disk, may be used. The amount of data transferred between disks during the repair is the repair bandwidth. Due to the upper limit of the network capacity, it is desirable to minimize the repair bandwidth in order to avoid reaching this limit. However, there is an inverse relationship between the repair bandwidth and the number of redundancy bits. 

Current distributed data storage systems utilize MDS codes \cite{CORE}. Let $\F_q$ be the field of $q$ elements, where $q$ is a power of a prime. An $[n,k,d]$ code $C$ over $\F_q$ is a linear subspace of $\F_q^n$ with dimension $k$ and minimum distance $d$.  The minimum distance $d(C)$ of a code $C$ is defined as 
	\[d(C):= \min\{ d(x,y): x,y \in C \quad \text{and} \quad x\neq y\}\]
where $d(x,y):=| \left\{ i : x_i \neq y_i \right\}|$ denotes the usual Hamming distance between $x,y \in \F_q^n$. 
A maximum distance separable (MDS) code is a linear code which achieves the Singleton bound, meaning $k=n-d+1$; otherwise, $k\leq n-d+1$. The rate of a code $C$ is defined to be $r= \frac{k}{n}$, and $n-k$ is the number of redundant bits. The parity check matrix $H$ of a code $C$ is an $(n-k) \times n$ matrix with rank $n-k$ such that $Hc=0$ for all codewords $c \in C$. Note that each row of $H$ consists of the coefficients of a parity check equation; that is, for row $h_i$ of $H$, $ 1 \leq i \leq n-k$,
	\[\sum_{j = 1}^n h_{ij}c_j \mod q =  0\]
for all $c \in C$. The code $C$ is a locally recoverable code with locality $r$ if and only if a received word $w$, coordinate $w_i$ can be recovered by accessing $w_{i_1}, \dots, w_{i_r}$ rather than all other coordinates of the received word $w_j$ where $j \neq i$ \cite{Local}.

The codes with the highest rate, i.e., the smallest number of redundancy bits, are MDS codes. Hence, a system with an underlying MDS code can recover encoded files using any $k$ disks and has the optimal amount of redundancy bits,  $n-k$ disks, but it requires a larger bandwidth than may be desired. There has been an ample amount of research to form a system using MDS codes which minimizes the recovery bandwidth with the minimal amount of storage or minimizes the storage space with the minimal bandwidth \cite{NetworkCoding, MSR, Rashmi, DiagonalParity, Suh}. In this paper, we focus on minimizing the bandwidth by considering codes which are not MDS from a graph theoretic viewpoint. 

A graph $G=(V,E)$ consists of a set of vertices $V$ and a set of edges $E$. A path $P_n$ is a sequence of $n$ distinct vertices $v_1, \dots, v_n$ such that $v_i$ and $v_{i+1}$ are adjacent for all $1 \leq i \leq n-1$. A cycle $C_n$ is a sequence of $n$ distinct vertices $v_1, \dots, v_n$ such that $v_i$ and $v_{i+1}$ are adjacent for all $1 \leq i \leq n-1$ and $v_n$ is adjacent to $v_1$. The girth $g$ of a graph $G$ is the length of the shortest cycle of graph $G$. The degree of a vertex $v\in G$, denoted $deg(v)$, is the number of edges incident to vertex $v$. If graph $G$ is a digraph, a directed graph, then the degree of a vertex $v$ is split into two parts: the outdegree, $\outdeg(v)$, which is the number of outgoing edges incident to vertex $v$; and the indegree, $\indeg(v)$, which is the number of incoming edges incident to vertex $v$. If every vertex $v \in G$ has degree $t$, then $G$ is called a $t$-regular graph. 

A connected graph $G$ with $n$ nodes and $m$ edges defines a $q$-ary code  $C_E(G)$ of length $m$ as follows.  Coordinate positions are associated with the edges of $G$, and the code $C_E(G)$ is spanned by vectors $w$ which correspond to spanning subgraphs of $G$, meaning those $w \in \F_q^n$ with  $$w_e= \begin{cases} 1 & e \in E(S) \\ 0 & \mathrm{otherwise} \end{cases}$$ 
for some spanning subgraph $S$ of $G$. The code $C_E(G)$ is $[m,m-n+1,g]$ code where $g$ is the girth of $G$ \cite[Theorem 10.11.5]{GraphBook}.

In this paper, we consider a distributed data storage system as a graph. A storage disk with $r$ blocks will be considered as a set of $r$ edges, where an edge represents a block of the storage disk and the vertices of the graph act as parity checks. We consider graphs which can be covered with paths on 4 vertices, and each such path is considered to be a disk.  In this model, we seek quick recovery of a disk with the minimum amount of communication on the graph.  This setup is detailed carefully in the next section. In Section \ref{section:SystemProps}, we study properties of this DSS, which is further demonstrated with an example Section \ref{section:ex}. Concluding discussion is in Section \ref{section:conclusion}.

\section{Set up of the Distributed Storage System}

\subsection{DSSs from $3$-regular graphs}
As was previously mentioned, we consider a distributed data storage
system as a graph. A storage disk with $r$ blocks is a set of $r$
edges, meaning each edge represents a block of the storage disk, and the
vertices of the graph act as parity checks. More specifically, we
consider disks to be comprised of three blocks, each block consisting of
an element from the field $\F_q$. In addition, we restrict the graphs
to the class of $3$-regular. Then the following theorem, proved independently by \cite{Kotzig} and \cite{BF}, may be applied.

\begin{theorem} \label{Theorem: PathsDecomp} \cite[Theorem 2.2]{PathsDecomp} Every $3$-regular graph can be
decomposed into disjoint $P_4$s, paths on three edges.
\end{theorem}

Let $\mathcal{G}$ be a $3$-regular graph.  Then $\mathcal{G}$ can be decomposed into $P_4$s, according to Theorem \ref{Theorem: PathsDecomp}. Each $P_4$ given by the decomposition represents a storage disk of the DSS. Note that the word disk can refer to either the 
physical device or to the $P_4$ path representing the physical device; the use of the word will be clear from the context.

For each vertex $v$ of the $3$-regular graph $\mathcal{G}$, any two incident edges determine a third edge which is also incident to $v$. Hence, the code obtained from the graph is a locally recoverable code with locality 2. This is demonstrated in the next example.

\begin{exmp}
Consider the Petersen graph $\mathcal{G}$ in Figure \ref{figure:petersengraph1}, a $3$-regular graph on 10 vertices with girth 5. By Theorem 
\ref{Theorem: PathsDecomp}, $\mathcal{G}$ can be decomposed into five disjoint $P_4$ paths. For instance, $\mathcal{G}$ can be decomposed into the five paths of three edges
as given in Figure \ref{figure:disks}.

\begin{figure}[htb]
  \centering

  \begin{subfigure}[b]{0.43\textwidth}
    \centering
    \includegraphics[scale=.65]{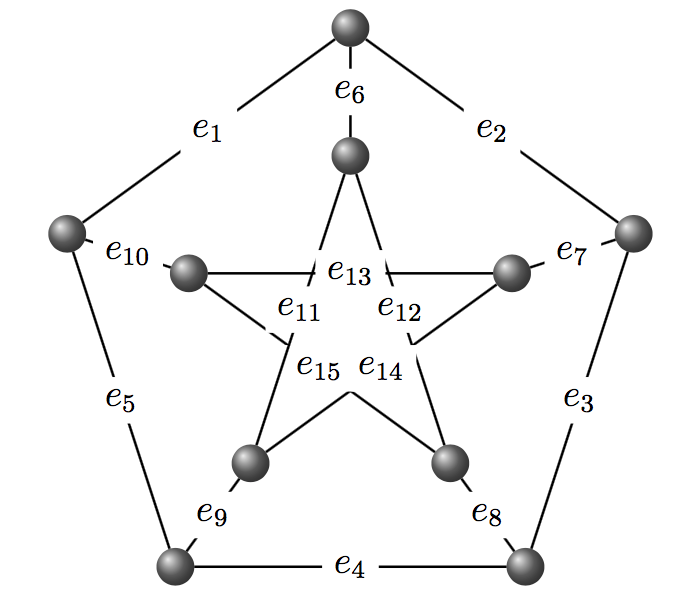}
		\caption{\label{figure:petersengraph1}}

                      \end{subfigure}
                      \begin{subfigure}[b]{0.43\textwidth}
                      \begin{center}
			\begin{tabular}{c|c|c|c|} \cline{2-4}
				 Disk 1: & $e_1$ & $e_{5}$ & $e_{6}$ \\ \cline{2-4}
				 Disk 2: & $e_7$ & $e_{10}$ & $e_{13}$ \\ \cline{2-4}
				 Disk 3: & $e_4$ & $e_9$ & $e_{14}$ \\ \cline{2-4}
				 Disk 4: & $e_2$ & $e_3$ & $e_8$ \\ \cline{2-4}
				 Disk 5: & $e_{11}$ & $e_{12}$ & $e_{15}$ \\ \cline{2-4}
			\end{tabular}  
                       \vspace{3cm}
                        \caption{\label{figure:disks}}
                      \end{center}
                      \end{subfigure}
                      \caption{Petersen graph.}
\end{figure}

Suppose edge $e_5$ is erased. Then one can recover this lost edge by one of the following parity checks:
	\[e_1+e_5+e_{10}=0 \quad \text{or} \quad e_4+e_5+e_9=0.\]

\end{exmp}

Given a distributed data storage system and its associated $3$-regular
graph $\mathcal{G}$, one goal is to maximize the number of disks that
can be repaired.  The restriction to $3$-regular graphs
  allow us to maximize the repair capability of the system while maintaining
  a small communication bandwidth. Note that the path
  decomposition of graphs with higher regularity is in general still
  an open problem in combinatorics \cite{Odile}.

We now consider the relationship between unrecoverable disks and cycles of the $3$-regular graph $\mathcal{G}$ representing the system.

\begin{theorem} \label{Theorem: ErasureCycle}
An edge of a graph $\mathcal{G}$ is not recoverable if and only if it is contained in a cycle of erased edges.
\end{theorem}
\begin{proof}
($\Rightarrow$) Let $\mathcal{G}$ be a $3$-regular graph. Suppose edge $e_i$ is not recoverable. Then for each vertex $v$ incident with the edge, there is at least one additional non-recoverable edge incident with $v$; otherwise, $e_i$ may be recovered via parity check. Similarly, for each of these non-recoverable edges, additional incident edges must be non-recoverable. Hence, there is a path of non-recoverable edges, say
	\[P: e_1 - e_2 - \dots - e_i - \dots - e_t.\]
Suppose the path $P$ is not a cycle. Then there must be an end vertex of $P$ with two known edges. As a result, one of the non-recoverable edges, say $e_t$, may be recovered, which is a contradiction. Hence, $P$ is a cycle and the edge $e_i$ is contained in a cycle of erased edges.

($\Leftarrow$) Suppose a cycle in $\mathcal{G}$ has been erased. Since each vertex has only three incident edges and two of these lie in the cycle, no edge of the cycle can be recovered.
\end{proof}

\begin{cor}
 A set of erased disks is not recoverable if and only if the corresponding edges form a cycle. Hence, a system can repair any $m$ disk erasures if the girth of the associated $3$-regular graph is at least $m+1$. 
\end{cor}

Motivated by the previous corollary, we next consider how to construct $3$-regular graphs with predetermined disk assignments from those that are $4$-regular.

\subsection{$3$-regular graphs with disk assignments from $4$-regular graphs}

Let $G$ be a $4$-regular graph on $N$ vertices. Then $G$
is Eulerian, since all vertices of $G$ have even degree. An Eulerian tour on $G$ 
prescribes a direction to each edge of $G$ as well as an ordering on the vertices of
$G$. Denote by $G_d=(V=\{v_1,\dots,v_N\},E)$ the directed
graph obtained from $G$ in this way.  The pair $(v_i,v_j)\in E$ denotes an edge from $v_i$
to $v_j$. Figure
\ref{fig:A} and \ref{fig:B} provides an example of a graph
$G$ with the corresponding digraph $G_d$.
Notice that $G_d$ is a $2$-regular directed graph, meaning
for each vertex $v \in V$, $\indeg(v)=\outdeg(v)=2$. The set of edges
incident to $v\in V$ are then 
$\In(v)=\{(v_{i_1},v), (v_{i_2},v)\}$ and
$\Out(v)=\{(v,v_{j_1}), (v,v_{j_2})\}$. It will be convenient to define functions $\min$ and $\max$ as follows:
$$\min \In(v) := \left( v_{\min\{ i_1,i_2\}}, v \right)$$ and $$\max \In(v) := \left( v_{\max \{ i_1,i_2\}}, v \right), $$
with analogous definitions for $\min \Out(v)$ and $\max \Out(v)$; that is, $\min$ and $\max$ are taken with respect to the subindices of the vertices.

\medskip
From $G_d$ we define an undirected graph $\mathcal{G}=(\mathcal{V}, \mathcal{E})$
where  $\mathcal{V}=E$, meaning that $V(\mathcal{G})=E(G_d)$, i.e., the vertices of $\mathcal{G}$ are
the edges of $G_d$, and $\{(v_{\ell_1},v_{m_1}),(v_{\ell_2},v_{m_2})\}\in
\mathcal{E}$ if and only if $(v_{\ell_1},v_{m_1})\neq (v_{\ell_2},v_{m_2})$, and
\begin{itemize}
\item either $v_{m_1}=v_{m_2}$, or 
\item $v_{\ell_r}=v_{m_s}=v$ for $r\neq s$ and
  \begin{itemize}
  \item either $(v_{\ell_s},v)=\min \In(v)$ and $(v,v_{m_r})=\min
    \Out(v)$, or 
  \item $(v_{\ell_s},v)=\max \In(v)$ and $(v,v_{m_r})=\max
    \Out(v)$.
  \end{itemize}

\end{itemize}

 Notice that the $\min$/$\max$ step of the construction of the graph
  $\mathcal{G}$ is a choice made for ease of notation. The
  properties of the distributed storage system that will be
  constructed are independent of this choice as will be explored in Section~\ref{section:SystemProps}.

\begin{exmp} 
Consider the $4$-regular complete bipartite graph $K_{4,4}$. Here, $G:=K_{44}$ is a graph with $N=8$ vertices and $2N=16$ edges. Figure \ref{fig:B} shows an Eulerian tour on $G$. By using the resultant directed graph $G_d$, we obtain the following set of edges
$$	 E = \left\{ 
\begin{array}{l}
(v_1,v_2), (v_1,v_6), (v_2,v_3),(v_2,v_7), (v_3,v_4),(v_3,v_8), (v_4, v_1), (v_4,v_5),\\(v_5,v_2), (v_5,v_6), (v_6,v_3),(v_6,v_7),  (v_7,v_4), (v_7,v_8),(v_8,v_1), (v_8,v_5) \end{array} \right\}.$$ 
This is the set of vertices for the associated $3$-regular graph $\mathcal{G}$ also, meaning  $V(\mathcal{G})=\mathcal{V}= E$. A set of neighboring vertices for the vertex $(v_1,v_2)$ is
	\[\{(v_4,v_1), (v_5,v_2), (v_2,v_3)\}\]
which comes from choosing the minimum $\mbox{In}(v_1)$ and the minimum $\mbox{Out}(v_2)$. Similarly, we may obtain a set of paths of length 3:

\bigskip
\begin{minipage}[h]{.49\linewidth}
 $ D_1: (v_1,v_2) - (v_4,v_1)- (v_8,v_1) - (v_1,v_6)$\\
 $ D_2: (v_2, v_3) - (v_1,v_2) - (v_5,v_2) - (v_2, v_7)$\\
 $ D_3: (v_3,v_4) - (v_2,v_3)- (v_6,v_3) - (v_3, v_8)$\\
 $ D_4: (v_4, v_1) - (v_3,v_4) - (v_7,v_4) - (v_4,v_5)$\\
\end{minipage}\hfill
          \begin{minipage}[h]{.49\linewidth}
            $D_5: (v_5, v_2) - (v_4,v_5) - (v_8,v_5) - (v_5, v_6)$\\
            $ D_6: (v_6, v_3) - (v_1,v_6) - (v_5,v_6) - (v_6, v_7)$\\
            $ D_7: (v_7, v_4) - (v_2,v_7) - (v_6,v_7) - (v_7,v_8)$\\
            $ D_8: (v_8, v_1) - (v_3, v_8) - (v_7,v_8) - (v_8, v_5)$.\\

        \end{minipage}

Again, note that the graph $\mathcal{G}$ is a $3$-regular graph.

\begin{figure}[htb]
  \centering

  \begin{subfigure}[b]{0.43\textwidth}
    \centering
    \includegraphics[scale=.65]{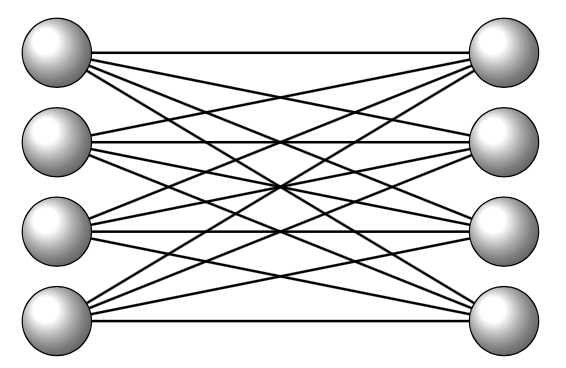}
    
      
        


        

    \caption{Graph $G$}
    \label{fig:A}
  \end{subfigure}
  \begin{subfigure}[b]{0.43\textwidth}
    \centering
    \includegraphics[scale=.65]{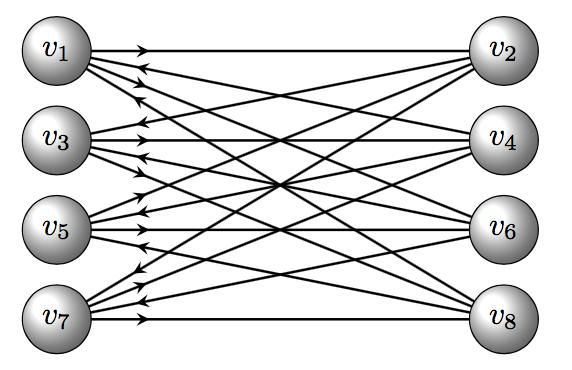}
    
      
        


        
    \caption{Graph $G_d$}
    \label{fig:B}
  \end{subfigure}

  \begin{subfigure}[b]{\textwidth}
    \centering

    \includegraphics[scale=.65]{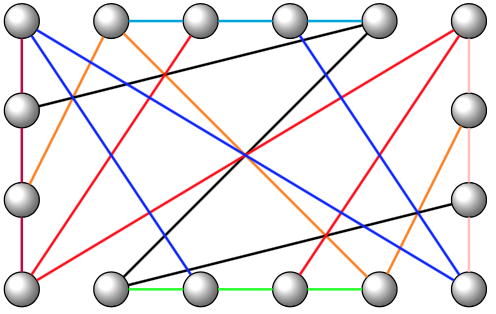}
    \caption{Graph $\mathcal{G}$}
    \label{fig:C}
  \end{subfigure}

\caption{Graph relations.\label{Figure: K44}}
\end{figure}

\end{exmp}

\begin{lm}
  The graph $\mathcal{G}$ obtained from the $4$-regular graph $G$ is $3$-regular.
\end{lm}
\begin{proof}
Consider a vertex $(v_\ell, v_m) \in \mathcal{V}$. Then
	\[Nbd\left[(v_\ell, v_m)\right] \subseteq \left\lbrace(v_t,v_m), (v_m, v_s), (v_r, v_\ell) \right\rbrace  \subseteq \mathcal{V} \]
for some $t,r,s$.
Note for fixed $\ell$ and $m$, $\ell \neq m$, $|\{(v_t,v_m) \in \mathcal{V}, t\neq \ell\}| =1$ and $|\{(v_m,v_s) \in \mathcal{V}\}| =|\{(v_r,v_\ell) \in \mathcal{V}\}| =2$. The construction requires one from the first set and only one from the latter two sets. Hence, $|Nbd\left[(v_\ell, v_m)\right] | = 3$.

\end{proof}

Recall that the constructed $3$-regular graph $\mathcal{G}$ represents a distributed
storage system where the $N$ disjoint $P_4$ paths correspond to disks and the edges correspond to blocks of storage within the disks. Hence, any $4$-regular graph $G$ with $N$ vertices gives rise to a distributed storage system with $N$ disks, each comprised of three blocks of storage. 

Next, we examine the recovery properties of such a system.

\section{Properties of the Distributed Storage System} \label{section:SystemProps}

Let $G$ be a $4$-regular graph and $\mathcal{G}$ an associated $3$-regular graph
as described in the previous section. Recall that the disks are the 
  $P_4$s decomposing the $3$-regular graph $\mathcal{G}$.

\subsection{Recovery procedure and performance}

As explained earlier, the edges of the $3$-regular graph $\mathcal{G}$
correspond to information symbols and the vertices to parity check
equations defining a linear code of $\F_q^N$. Valid realizations of
the distributed storage system correspond to codewords of such a code.

Next, we construct a minimum distance decoder which is sequential in nature
and reduces the communication bandwidth based on the fact that disks
are paths.

\begin{lm}
  The required communication bandwidth  for recovering a disk of the graph
  $\mathcal{G}$ is 4 $\F_q$--symbols.
\end{lm}

\begin{proof}
  Let $D$ be a disk in $\mathcal{G}$. Figure \ref{Figure: disk}
  depicts the subgraph of $\mathcal{G}$ containing the disk $D$ and
  its adjacent edges.


\begin{figure}[h]
\begin{center}






  \includegraphics[scale=.65]{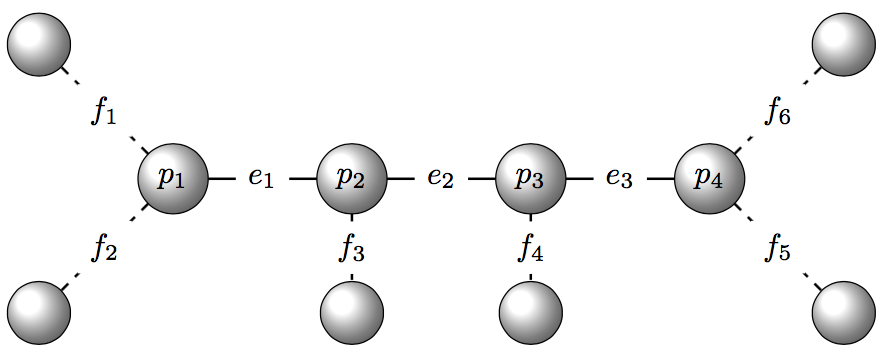}
 \caption{\label{Figure: disk} The disk $D$ with adjacent edges.}
 \end{center} 
\end{figure}

The associated parity check equations given by the graph are
  \begin{eqnarray*}
    \label{eq:1}
    e_1+f_1+f_2&=&0,\\
    e_1+e_2+f_3&=&0, \\
    e_2+e_3+f_4&=&0 \\
    e_3+f_5+f_6&=&0.
  \end{eqnarray*}

  These conditions are redundant as it is possible to recover $e_1$,
  $e_2$ and $e_3$ from any choice of 3 equations. In order to minimize
  the bandwidth we use the first three equations. Hence, in order to
  recover the disk $D$ it is enough to access $f_1$, $f_2$, $f_3$
  and $f_4$.
\end{proof}

The recovery procedure requires three rounds of computation. For instance,
in order to recover $e_3$ we need $e_2$, but in order to recover $e_2$
we need $e_1$. Note that if the disks are chosen such that the blocks
are not in a path, the required communication bandwidth would be higher. 

It is possible to consider the trade-off between the rounds of
communication and the bandwidth. One can either communicate four $\F_q$--symbols and use three computation rounds, or communicate five $\F_q$--symbols
and use two computation rounds.  The choice can be based on performance
of communication versus computation.

\begin{cor}
  If $n$ disconnected disks are erased, then $4n$ $\F_q$--symbols is required to recover the disks.
\end{cor}

\subsection{Recovery bound}

To determine an upper bound on the number of recoverable disk erasures of a DSS, we 
study the underlying $4$-regular graph.

\begin{defn}
 A $t$-disk cycle of $\mathcal{G}$ is 
 a cycle of $\mathcal{G}$ whose edges belong to $t$ disks.
\end{defn}

\begin{prop} \label{Prop: cycle}
The $3$-regular graph $\mathcal{G}$ contains a $t$-disk cycle if and only if the associated $4$-regular graph $G$ contains a cycle on $t$ vertices.
\end{prop}
\begin{proof}
  ($\Leftarrow$) Suppose there is a cycle of $t$ vertices in the
  $4$-regular graph $G$, and let
    $v_1, \ldots, v_t$ be those vertices. Recall that  each vertex of
    $G$ corresponds to a disk in $\mathcal{G}$ and each edge
  of $G$ corresponds to a vertex in $\mathcal{G}$. Let $e_i$ be
    the edge between $v_i$ and $v_{i+1}$.  Then in $\mathcal{G}$,
  vertex $e_i$ is connected to vertex $e_{i+1}$ via an edge that is a
  part of disk $v_{i+1}$. Since $e_i$ is connected to $e_{i+1}$,
  $1 \leq i \leq t$, and $e_1$ is connected to $e_t$, we have a
  $t$-disk cycle in $\mathcal{G}$.
	
  ($\Rightarrow$) Suppose there exists a $t$-disk cycle $C$ in
  $\mathcal{G}$. Let $v_1, \ldots, v_t$ be the disks of such
  a cycle. Recall that each disk is a vertex of the $4$-regular
  graph $G$. If disks $v_i$ and $v_j$ are adjacent in $\mathcal{G}$,
  then edge $\{v_i,v_j\} \in E(G)$. It follows that vertices
    $v_1,\dots,v_t$ together with these edges form a cycle in $G$ of
    length $t$.
\end{proof}

We immediately notice a relationship between a minimal disk cycle of $\mathcal{G}$ and the girth of the $4$-regular graph $G$.

\begin{cor}
The smallest $t$ for which $\mathcal{G}$ has a $t$-disk cycle is $t=\mbox{girth}(G)$.
\end{cor}



Using the previous two results, we now show that a DSS with an underlying $4$-regular graph
of girth $g$ allows for the recovery of any $g-1$ disk erasures.
\begin{theorem} \label{4regthm}
Let $G$ be a $4$-regular graph with girth $g$. Then 
	\begin{itemize}
		\item any $g-1$ disk erasures can be recovered in the corresponding $3$-regular graph
$\mathcal{G}$, and 
		\item there exists an erasure pattern of $g$ disk erasures that cannot be recovered. 
	\end{itemize}
\end{theorem}
\begin{proof}
Suppose $g-1$ disks have been erased and that one of the disks is not recoverable. By Theorem \ref{Theorem: ErasureCycle}, there is a cycle containing a portion of this disk that has been erased. This cycle is then a $t$-disk cycle, with $t \leq g-1$, which is a contradiction according to Proposition \ref{Prop: cycle}. Hence, the $g-1$ disks can be recovered.

To show there exists an erasure pattern of $g$ disk erasures that cannot be recovered, consider the cycle of $G$ with length $g=\mbox{girth}(G)$. Consider the corresponding $g$-disk cycle of the $3$-regular graph $\mathcal{G}$. Erasure of these $g$ disks would lead to an unrecoverable pattern of erasures.
\end{proof}

\subsubsection*{Optimality of cage graphs.}
According to Theorem \ref{4regthm}, the number of recoverable disk
erasures depends only on the girth of the $4$-regular graph. Hence, to
maximize the amount of recoverable data, we consider the $4$-regular
graphs $G$ with the minimum number of nodes given a fixed girth
$g$. Such graphs are called cage graphs. 

More precisely, a
$(k,g)$-cage graph $G$ is a $k$-regular graph of girth $g$ with the
minimum possible number of nodes. The study of cage graphs was initiated by Tutte \cite{Tutte}. Although it is known that $(k,g)$-cage graphs exist for all $k \geq 2$ and $g\geq 3$, few constructions are known. Some results are known for $k=4$, which leads us to consider
distributed storage systems from $(4,g)$-cage graphs. Specifically, the complete graph $K_5$ is a $(4,3)$-cage while the complete bipartite graph is a $(4,4)$-cage \cite{wong}; the Robertson graph which is a graph on $19$ vertices is a $(4,5)$-cage \cite{Rob}; the 
point-line incidence graph of $PG(2,3)$, which is a graph on $26$ vertices is a $(4,6)$-cage; 
and there is a $(4,7)$ cage with $67$ vertices \cite{Exoo47}. See \cite{CageSurvey}  for a survey of progress on this open problem. 

In Table \ref{tab:table1}, we consider the codes $C_E(\mathcal{G})$ where $\mathcal{G}$ is obtained from a $(4,g)$-cage graph $G$. If $G$ has $N$ vertices, then $\mathcal{G}$ has $3N$ edges and $2N$ vertices. Hence,  $C_E(\mathcal{G})$ is a $[3N, N,g]$ code. 

\begin{table}[h!]
  \centering
  \begin{tabular}{|c|c|c|c|c|}
  \hline
   \# disks & \# blocks  & \# disks recoverable & \# blocks recoverable & code parameters\\
   \hline
   5 & 15 & 2 & 6 & [15, 6, 3] \\
   \hline 
   8 & 24 & 3 & 9 & [24, 9, 4] \\
   \hline
   19 & 57 & 4 & 12 & [57, 20, 5] \\
   \hline
   26 & 78 & 5 & 15 & [78, 27, 6] \\
   \hline
   67 & 201 & 6 & 18 & [201, 68, 7]\\
   \hline
  \end{tabular}
\caption{}
  \label{tab:table1}
\end{table}

For purposes of comparison, we mention the best known comparable binary linear codes \cite{Grassl} (based on the initial version of Brouwer's tables \cite{Brouwer}) meaning codes with the largest minimum distance among all known codes with the same length and dimension. These are $[15, 6, 6]$, $[24, 9, 8]$, $[57, 20, 16 \leq d \leq 18]$, $[78,27,20 \leq d \leq 26]$, and $[201, 68, 40 \leq d \leq 62]$. Keep in mind that while an $[n,k,d]$ code can correct any $\left\lfloor \frac{d-1}{2} \right\rfloor$ errors, it does so with access to the entire received word (meaning access to every node) as opposed to using just local information from selected nodes.

Recall that given a connected graph $G$ with $m$ edges on $n$ vertices with a girth of $g$, the code $C_E(G)$ has rate $\frac{m-n+1}{m} = 1-\frac{n-1}{m}$. Hence, for any $3$-regular graph with $n$ vertices, the rate is the function $R(n) = 1- \frac{n-1}{3n/2}$. Any $3$-regular graph constructed from the above $4$-regular graph on five vertices has a rate of $\frac{2}{5}$ and can recover any two erasures. Hence, the distributed data storage systems based on any of the resulting $3$-regular graphs allows for the same rate and the recovery of the same number of disk erasures. Note that $R(n)$ is a decreasing function with a lower bound of $\frac{1}{3}$. Hence, the rate improves by decreasing the number of vertices, which decreases the number of disks. As a result, the best data distributed storage system in terms of rate is the system consisting of five disks. 

\section{Example: The System with 5 Disks} \label{section:ex}

\begin{figure}[htb]
  \centering

  \begin{subfigure}[b]{0.32\textwidth}
    \centering
			
			
			

			
			
			

     \includegraphics[scale=.65]{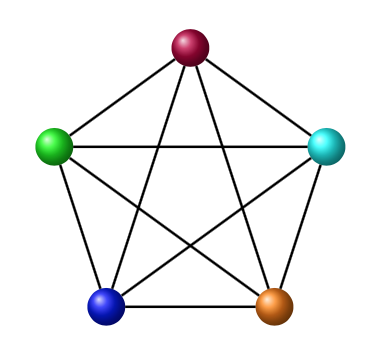}
    \caption{Graph $G$ \label{Figure: K_5a}}
    \label{fig:AA}
  \end{subfigure}
  \begin{subfigure}[b]{0.32\textwidth}
    \centering
    
			
			
			

			
			
			
			
			
                      \includegraphics[scale=.65]{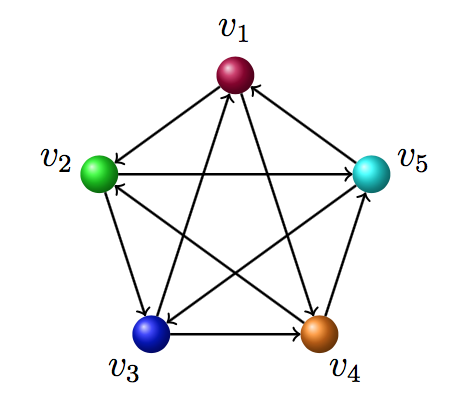}
			\caption{\label{Figure: K_5b} Graph $G_d$}
    \label{fig:BB}
  \end{subfigure}
  \begin{subfigure}[b]{0.32\textwidth}
    \centering
\includegraphics[scale=.68]{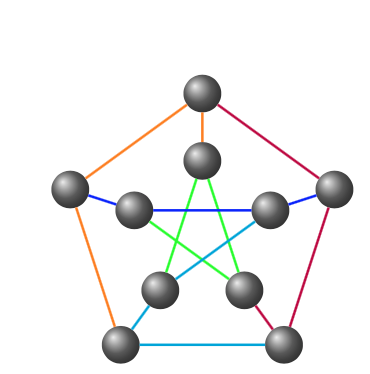}
    \caption{Graph $\mathcal{G}$}
    \label{fig:CC}
  \end{subfigure}

\caption{Graph relations.\label{Figure: K5}}
\end{figure}


 
 Suppose we want a distributed data storage system with
  five disks. Then we consider the $4$-regular graph $G$ on five vertices,
  which is the complete graph $K_5$. Figure \ref{fig:BB} depicts an Eulerian
  tour on $G$. The edges of the
  resultant directed graph $G_d$ are
	\[ E = \{(v_1,v_2),(v_1.v_4),(v_2,v_3),(v_2,v_5),(v_3,v_1),(v_3,v_4),(v_4,v_2),(v_4,v_5),(v_5,v_1),(v_5,v_3)\},\]
which then gives the vertices of the associated $3$-regular graph $\mathcal{G}=(\mathcal{V}, \mathcal{E})$, meaning 
$\mathcal{V} = E$. The disks which are disjoint $P_4$ paths of $\mathcal{G}$ can be chosen to be
	\begin{align*}
		& D_1: (v_1,v_2) - (v_5,v_1) - (v_3,v_1) - (v_1,v_4)\\
		& D_2: (v_2,v_3) - (v_1,v_2) - (v_4,v_2) - (v_2,v_5) \\
		& D_3: (v_3,v_1) - (v_5,v_3) - (v_2,v_3) - (v_3,v_4) \\
		& D_4: (v_4,v_2) - (v_1,v_4) - (v_3,v_4) - (v_4,v_5) \\
		& D_5: (v_5,v_1) - (v_4,v_5) - (v_2,v_5) - (v_5,v_3) 
	\end{align*}
resulting in the Petersen graph, which has girth 5.

        Since the $4$-regular graph $G$ has  girth 3, from Theorem
        \ref{4regthm} any two erased disks can be recovered. Hence, the 
        constructed Petersen graph 
        allows for the recovery of any four block erasures and any two disk erasures.

        Note that having a 3-disk cycle does not imply that
        $\mathcal{G}$ has girth 3.
        More importantly, recall the minimum disks cycle is a
        3-disk cycle due to the construction of $\mathcal{G}$
        and hence is independent of the resulting girth. The girth of
        the constructed $3$-regular graph $\mathcal{G}$ depends on the
        choice of the placement of vertices in the disk.  In this
        example, the resulting girth has a lower bound of 3 and an
        upper bound of 5. 

A graph with a girth of $3$ can be obtained by choosing the following paths:
	\begin{align*}
		& D_1: (v_1,v_2) - (v_3,v_1) - (v_5,v_1) - (v_1,v_4)\\
		& D_2: (v_2,v_3) - (v_1,v_2) - (v_4,v_2) - (v_2,v_5) \\
		& D_3: (v_3,v_1) - (v_2,v_3) - (v_5,v_3) - (v_3,v_4) \\
		& D_4: (v_4,v_2) - (v_1,v_4) - (v_3,v_4) - (v_4,v_5) \\
		& D_5: (v_5,v_1) - (v_2,v_5) - (v_4,v_5) - (v_5,v_3) 
	\end{align*}
The associated system allows for a recovery of up to any two erasures, the same number as 
the Petersen graph.

With this example, we see that the girth of the resultant graph does not play a role in the number of recoverable erasures, 
and that it is independent of the rate of the code based on the graph.


\section{Conclusions} \label{section:conclusion}

In this paper, we construct distributed data storage systems from
3-regular graphs. The obtained DSSs have the capability to recover any
maximal number of erased disks through the connectivity of the
underlying graph by a sequential decoder. Due to the correspondence
between the distributed storage disks and the paths of the graph, we are
able to keep a small communication bandwidth.


%

\newcommand*{\innervertices}[5]{
  \drawvertices{#1}{#2}{#3}{#4}{#5}{\RA}{a}
}

\newcommand*{\outervertices}[5]{
  \drawvertices{#1}{#2}{#3}{#4}{#5}{\RB}{b}
}

\newcommand*{\staredges}[5]{
  \foreach \u/\v/\edgecolor/\edgelabel in {1/3/#1/15, 3/5/#2/12, 5/2/#3/11, 2/4/#4/14, 4/1/#5/13} {
    \tikzstyle{EdgeStyle} = [color = \edgecolor]
    \Edge[label = $\edgelabel$](a\u)(a\v)
  }
}

\newcommand*{\outeredges}[5]{
  \foreach \u/\v/\edgecolor/\edgelabel in {1/2/#1/5, 2/3/#2/4, 3/4/#3/3, 4/5/#4/2, 5/1/#5/1} {
    \tikzstyle{EdgeStyle} = [color = \edgecolor]
    \Edge[label = $\edgelabel$](b\u)(b\v)
  }
}

\newcommand*{\betweenedges}[5]{
  \foreach \u/\edgecolor/\edgelabel in {1/#1/10, 2/#2/9, 3/#3/8, 4/#4/7, 5/#5/6} {
    \tikzstyle{EdgeStyle} = [color = \edgecolor]
    \Edge[label= $\edgelabel$](a\u)(b\u)
  }
}


\end{document}